\documentclass[a4paper,UKenglish]{lipics-v2019}
\usepackage{latexsym}
\usepackage{amsmath}
\usepackage{verbatim}
\usepackage{xcolor}
\usepackage{enumerate}
\usepackage{amssymb}
\usepackage{bbm}
\usepackage{algpseudocode}
\usepackage{graphicx}
\usepackage{textcomp}
\usepackage[ruled,vlined,algosection,linesnumbered]{algorithm2e}

\usepackage{tikz}
\usetikzlibrary{decorations.shapes,decorations.pathreplacing,decorations.pathmorphing}
\usetikzlibrary{arrows,matrix,shapes,snakes}
\usetikzlibrary{positioning}
\tikzset{
        stars/.style={star,inner sep=2pt}
    }

\usepackage{epsfig}
\usepackage{amsfonts}
\usepackage{todonotes}
\usepackage[inline]{enumitem}
\usepackage{longtable}
\usepackage{url}
\usepackage{tabularx}
\usepackage{wasysym}
\usepackage{xspace}
\usepackage{varwidth}
\usepackage{stmaryrd}

\usepackage{hyphenat}

\newcommand{\NPC}{\textsf{NP}-complete\xspace}

\newcommand{\NPH}{\textsf{NP}-hard\xspace}

\theoremstyle{plain}
\newtheorem{observation}[theorem]{Observation}

\hypersetup{
hidelinks,
colorlinks=true,
citecolor=[rgb]{1 0 0},
linkcolor=[rgb]{1 0 0},
urlcolor=[rgb]{1 0 0}
}

\usepackage{thmtools, thm-restate,hyperref}
\newcommand{\old}[1]{{}}
 \usepackage{framed,xspace}
 \newcommand{\stree}{{\sc Steiner Tree}\xspace}
 \newcommand{\vertconst}{$24$ \xspace}
\newcommand{\edgeconst}{$24$ \xspace}
\newcommand{\cellconst}{$24$ \xspace}
\newcommand{\pwidth}{$7$ \xspace}
\newcommand{\pwidthtwo}{$6$ \xspace}
\newcommand{\exst}{{\sc Exact Steiner Tree}\xspace}
\newcommand{\compconst}{$24$ \xspace}

\hideLIPIcs
 \definecolor{shadecolor}{gray}{0.9}
 
\usepackage{microtype}

\title{Parameterized Study of \stree on Unit Disk Graphs}

\titlerunning{\stree on Unit Disk Graphs}

\author{Sujoy Bhore}{Algorithms and Complexity Group, TU Wien, Vienna,  Austria}{sujoy@ac.tuwien.ac.at}{https://orcid.org/0000-0003-0104-1659}{}
\author{Paz Carmi}{Ben-Gurion University of the Negev, Beersheba, Israel}{carmip@cs.bgu.ac.il}{https://orcid.org/0000-0003-0154-5013}{}
\author{Sudeshna Kolay}{Indian Institute of Technology Kharagpur, India}{skolay@cse.iitkgp.ac.in}{https://orcid.org/0000-0002-2975-4856}{}
\author{Meirav Zehavi}{Ben-Gurion University of the Negev, Beersheba, Israel}{meiravze@bgu.ac.il}{https://orcid.org/0000-0002-3636-5322}{}

\authorrunning{S. Bhore, P. Carmi, S. Kolay, and M. Zehavi}

\Copyright{Sujoy Bhore, Paz Carmi, Sudeshna Kolay and Meirav Zehavi}

\ccsdesc[500]{Theory of computation~Computational geometry}
\ccsdesc[300]{Theory of computation~Parameterized Algorithms}

\keywords{Unit Disk Graphs; \textsf{FPT}; Subexponential exact algorithms; \textsf{NP}-Hardness; \textsf{W}-Hardness}

\funding{Research of Sujoy Bhore was supported by the Austrian Science Fund (FWF), grant P 31119. Research of Paz Carmi was partially supported by the Lynn and William Frankel Center for Computer Science and by Grant 2016116 from the United States-Israel Binational Science Foundation. Research of Meirav Zehavi was supported by the Israel Science Foundation (ISF) grant no. 1176/18 and United States - Israel Binational Science Foundation (BSF) grant no. 2018302.}

\acknowledgements{We are grateful to the anonymous reviewers for their helpful comments.}

\nolinenumbers

\EventEditors{}
\EventNoEds{2}
\EventLongTitle{
}
\EventShortTitle{}
\EventAcronym{}
\EventYear{}
\EventDate{}
\EventLocation{}
\EventLogo{}
\SeriesVolume{}
\ArticleNo{}

\begin{document}

\maketitle

\begin{abstract}
We study the \stree problem on unit disk graphs. 
Given a $n$ vertex unit disk graph $G$, a subset $R\subseteq V(G)$ of $t$ vertices and a positive integer $k$, the objective is to decide if there exists a tree $T$ in $G$ that spans over all vertices of $R$ and uses at most $k$ vertices from $V\setminus R$. 
The vertices of $R$ are referred to as \emph{terminals} and the vertices of $V(G)\setminus R$ as \emph{Steiner} vertices. First, we show that the problem is \NPH.
Next, we prove that the \stree problem on unit disk graphs can be solved in $n^{O(\sqrt{t+k})}$ time. We also show that the \stree problem on unit disk graphs parameterized by $k$ has an FPT algorithm with running time $2^{O(k)}n^{O(1)}$. In fact, the algorithms are designed for a more general class of graphs, called clique-grid graphs~\cite{fomin2019finding}. We mention that the algorithmic results can be made to work for \stree on disk graphs with bounded aspect ratio. Finally, we prove that \stree on disk graphs parameterized by $k$ is W[1]-hard. 
\end{abstract}

\section{Introduction}\label{sec:intro}
Given a graph $G$ with a weight function $w:E(G)\rightarrow \mathbb{R}^+$
and a subset $R\subseteq V(G)$ of vertices, a Steiner tree is an acyclic subgraph of $G$ 
spanning all vertices of $R$. The vertices of $R$ are usually referred to as \emph{terminals} and the vertices of $V(G)\setminus R$ as \emph{Steiner} vertices. 
The {\sc Minimum Steiner Tree} problem is to find a Steiner tree $T$ such the total weight of $E(T)$ is minimized. The decision version of this is the  \stree problem, where given a graph $G$, a subset $R\subseteq V(G)$ of vertices and a positive integer $k$, the objective is to determine if there exists a Steiner tree $T$ in $G$ for the \emph{terminal} set $R$ such that the number of \emph{Steiner} vertices in $T$ is at most $k$. The \stree problem is one of Karp's classic \NPC problems \cite{k-racp-72}; moreover, that makes the optimization problem \NPH. 

A special case of the {\sc Minimum Steiner Tree} problem is the \textsc{Metric Steiner Tree} problem. Given a complete graph $G=(V,E)$, each vertex corresponds to a point in a metric space, and for each edge $e\in E$ the weight $w(e)$ corresponds to the distances in the space. In other words, the edge weights satisfy the triangle inequality. It is well known that, given an instance of the non-metric Steiner tree problem, 
it is possible to transform it in polynomial time into an equivalent instance of the \textsc{Metric Steiner Tree} problem. Moreover, this transformation preserves the approximation factor \cite{vazirani2013approximation}. 
The {\sc Euclidean Steiner Tree} problem or {\sc Geometric Steiner Tree} problem takes as input $n$ points in the plane. The objective is to connect them by lines of minimum total length in such a way that any two points may be interconnected by line segments either directly or via other points and line segments. The {\sc Minimum Steiner Tree} problem is \NPH even in Euclidean or Rectilinear metrics \cite{garey1977rectilinear}. 

Arora \cite{arora1998polynomial} showed that the {\sc Euclidean Steiner Tree}
and {\sc Rectilinear Steiner Tree} problems can be efficiently approximated arbitrarily close to the optimal. Several approximation schemes have been proposed over the years on {\sc Minimum Steiner Tree} for graphs with arbitrary weights \cite{berman1994improved, borchers1997thek, karpinski1997new, promel2000new}.  
Although the Euclidean version admits a PTAS, 
it is known that the \textsc{Metric Steiner Tree} problem is APX-complete. There is a polynomial-time algorithm that approximates the minimum Steiner tree to within a factor of $\ln(4)+\epsilon \approx 1.386$ \cite{chlebikova2008steiner}; however, approximating within a factor $\frac{96}{95} \approx 1.0105$ is \NPH \cite{berman20091}. 

The decision version, \stree is well-studied in parameterized complexity. A well-studied parameter for the \stree is the number of terminals $t = |R|$.
It is known that the \stree is FPT for this parameter 
due to the classical result of Dreyfus and Wagner \cite{dreyfus1971steiner}. Fuchs et al.~\cite{FuchsKMRRW07} and Nederlof~\cite{Nederlof13} gave alternative algorithms for \stree parameterized by $t$ with running times that are not comparable with the Dreyfus and Wagner algorithm. On the other hand, \stree parameterized by the number of Steiner vertices $k$ is W[2]-hard~\cite{downey2012parameterized}. 
Hence, the focus has been on designing parameterized algorithms for graph subclasses like planar graphs \cite{jones2013parameterized}, $d$-degenerate graphs \cite{suchy2017extending}, etc. 
In \cite{dvovrak2017parameterized}, Dco\v{r}\'{a}k et al. designed an efficient parameterized approximation scheme (EPAS) for the \stree parameterized by $k$ \footnote{For any $\epsilon>0$ computes a $(1+\epsilon)$ approximation in time $f(p,\epsilon)\times n^{O(1)}$ for a computable function $f$ independent of $n$.}. 

In this paper, we study the \stree problem on unit disk graphs when the parameter is the number of Steiner vertices $k$. Unit disk graphs are the geometric intersection graphs of unit circles in the plane. That is, given $n$ unit circles in the plane, we have a graph $G$ where each vertex corresponds to a circle such that there is an edge between two vertices when the corresponding circles intersect. Unit disk graphs have been widely studied in computational geometry and graph algorithms due to their usefulness in many real-world problems, e.g., optimal facility location \cite{wang1988study}, wireless and sensor networks; see \cite{hale1980frequency, kammerlander1984c}. These led to the study of many \NPC problems on unit disk graphs; see \cite{clark1991unit, dumitrescu2011minimum}.

There are some works on variants of {\sc Minimum Steiner Tree} on unit disk graphs in the approximation paradigm. Li et al. \cite{li2009ptas} studied node-weighted Steiner trees on unit disk graphs, and presented a PTAS when the given set of vertices is $c$-local. Moreover, they used this to solve the node-weighted connected dominating set problem in unit disk graphs and obtained a 
$(5+\epsilon)$-approximation algorithm. In \cite{biniazfull}, Biniaz et al. studied the {\sc Full Steiner Tree}\footnote{A full Steiner tree is a Steiner tree which has all the terminal vertices as its leaves} problem on unit disk graphs. They presented a $20$-approximation algorithm for this problem, and for $\lambda$-precise graphs gave a $(10+\frac{1}{\lambda})$-approximation algorithm where $\lambda$ is the length of the longest edge. Although there have been a plethora of work on variants of the {\sc Minimum Steiner Tree} problem on unit disk graphs in approximation algorithms, hardly anything is known in parameterized complexity for the decision version. 
In this regard, we refer to the work of Marx et al.~\cite{marx2018subexponential}
who investigated the parameterized complexity of the {\sc Minimum Steiner Tree} problem on planar graphs, where the number of terminals ($k$) is regarded as the parameter. They have designed an $n^{O(\sqrt{k})}$-time exact algorithm,
and showed that this problem on planar graphs
cannot be solved in time $2^{o(k)}\cdot n^{O(1)}$, assuming ETH. However, these results do not directly apply on unit disk graphs as unit disk graphs can contain very large cliques, but, then planar graphs contains arbitrarily large stars. Recently, Berg et al.~\cite{de2018framework} showed that the \stree problem can be solved in $2^{O(n^{1-\frac{1}{d}})}$ time on intersection graphs of $d$-dimensional similarly-sized fat objects, for some $d \in \mathbb{Z}_{+}$.

More often than not, the geometric intersection graph families such as unit disk graphs, unit square intersection graphs, rectangle intersection graphs, provide additional geometric structure that helps to generate algorithms. In this paper, our objective is to understand parameterized tractability landscape of the \stree problem on unit disk graphs.   

\paragraph*{Our Results.}
First in Section~\ref{sec:nph}, we show that \stree on unit disk graphs is \NPH. Then, in Section~\ref{sec:subexp}, we design a subexponential algorithm for the \stree problem on unit disk graphs parameterized by the number of terminals $t$ and the number of Steiner vertices $k$.

\begin{restatable}
{theorem}{subexp}\label{thm:subexp-UDG}
\stree on unit disk graphs can be solved in $n^{O(\sqrt{t+k})}$ time.
\end{restatable}

The approach to design this subexponential algorithm is very similar to that used in~\cite{fomin2019finding}. First, we apply a Baker-like shifting strategy to create a family $\mathcal{F}$ of instances (of \exst, which is a variant of \stree) such that if the input instance $(G,R,t,k)$ is a yes-instance then there is at least one constructed instance in $\mathcal{F}$ that is a yes-instance of \exst. On the other hand, if $(G,R,t,k)$ is a no-instance of \stree, then no instance of $\mathcal{F}$ is a yes-instance of \exst. With the knowledge that the answer is preserved in the family $\mathcal{F}$, we design a dynamic programming subroutine to solve \exst on each of the constructed instances of $\mathcal{F}$. 

Next, in Section~\ref{sec:FPT}, we show that the \stree on unit disk graphs has an FPT algorithm when parameterized by $k$.

\begin{restatable}
{theorem}{FPT}\label{thm:FPT-UDG}
\stree on unit disk graphs can be solved in $2^{O(k)}n^{O(1)}$ time.
\end{restatable}

Here, we show that solving the \stree problem on an instance $(G,R,t,k)$ is equivalent to solving the problem on an instance $(G',R',t',k)$ where the graph $G'$ is obtained by contracting all connected components of $G[R]$. Although $G'$ loses all geometric properties, we show that the number of terminals in $R'$ is only dependent on $k$. This essentially changes the problem to running the Dreyfus-Wagner algorithm on $(G',R',t',k)$.

Both the results in Theorem~\ref{thm:subexp-UDG} and \ref{thm:FPT-UDG} are shown to work for a superclass of graphs, called clique-grid graphs. We would like to remark that the algorithms can also be made to work for disk graphs with constant aspect ratio.

Finally, in contrast, in Section~\ref{sec:whard} we prove that the \stree problem for disk graphs is W[1]-hard, parameterized by the number Steiner vertices $k$. The \stree problem is known to be W[2]-hard on general graphs~\cite{downey2012parameterized}. 
However, it is not clear how to use that reduction for disk graphs. We show a reduction of our problem from {\sc Grid Tiling} with $\ge$~\cite{cygan2015parameterized},
ruling out the possibility of a $f(k)n^{o(k)}$ time algorithm for any function $f$, assuming ETH.

\begin{restatable}
{theorem}{whardness}\label{w1-hard}
\label{whard}
The \stree problem on disk graphs is W[1]-hard, parameterized by the number of Steiner vertices $k$.
\end{restatable}

\section{Preliminaries}\label{sec:prelims}

The set $\{1,2,\ldots,n\}$ is denoted as $[n]$. For a graph $G$, and a subset $V' \subseteq V(G)$, $G[V']$ denotes the subgraph induced on $V'$. The \exst problem takes as input a graph $G$, a terminal set $R$ with $t$ terminals and a positive integer $k$. The aim is to determine whether there is a Steiner tree $T$ in $G$ for $R$ that has exactly $k$ Steiner vertices. A Steiner tree with at most $k$ Steiner vertices is called a $k$-Steiner tree while one with exactly $k$ Steiner vertices is called an exact $k$-Steiner tree. Note that if $T$ is an exact $k$-Steiner tree then $\vert V(T) \vert = t+k$. When the \stree or \exst problem is restricted to taking input graphs only from a graph class $\mathcal{G}$, then these variants are referred to as \stree on $\mathcal{G}$ and \exst on $\mathcal{G}$, respectively.

\begin{observation}\label{obs:exact}
A tree $T$ is a $k$-Steiner tree for an instance $(G,R,t,k)$ if and only if $T$ is an exact $k'$-Steiner tree for the instance $(G,R,t,k')$ of \exst for some $k' \leq k$.
\end{observation}

\begin{definition}~\cite{fomin2019finding}\label{def:clique-grid}
 A graph $G$ is a clique-grid graph if there is a pair $p,p'\in \mathbb{N}$ and a function $f: V(G) \rightarrow [p]\times [p']$ such that the following conditions hold:
 \begin{enumerate}
  \item For all $(i,j) \in [p] \times [p']$, $f^{-1}(i,j)$ is a clique in $G$.
  \item For all $uv \in E(G)$, if $f(u) = (i,j)$ and $f(v) = (i',j')$ then $\vert i-i'\vert \leq 2$ and $\vert j-j' \vert \leq 2$.
 \end{enumerate}
 Such a function $f$ is called a representation of the graph $G$.
\end{definition}

Unit disk graphs are clique-grid graphs~\cite{fomin2019finding}. Next, we define a representation of a clique-grid graph called a cell graph.
\begin{definition}~\cite{fomin2019finding}\label{def:cell-graph}
 Given a clique-grid graph $G$ with representation $f:V(G) \rightarrow [p]\times [p']$, the cell graph ${\sf cell}(G)$ is defined as follows:
 \begin{itemize}
  \item $V({\sf cell}(G)) = \{v_{ij} \vert i \in [p], j \in [p'], f^{-1}(i,j)\neq \emptyset\}$,
  \item $E({\sf cell}(G)) = \{v_{ij}v_{i'j'} \vert (i,j) \neq (i',j'), \exists u \in f^{-1}(i,j) \mbox{ and } \exists v \in f^{-1}(i',j') \mbox{ such that } uv\in E(G)\}$.   
 \end{itemize}
\end{definition}
For each vertex $v_{ij}\in V({\sf cell}(G))$, the pair $(i,j)$ is also called a cell of $G$ and by definition corresponds to a non-empty clique of $G$. A vertex $v\in V(G)$ is said to be in the cell $(i,j)$ if $f(v) = (i,j)$. The neighbour of a cell $\mathcal{C} = (i,j)$ in a cell $\mathcal{C}' = (i',j') \neq \mathcal{C}$ are $\{v \in V(G) \vert f(v) = (i',j'), \exists u \mbox{ such that } f(u) = (i,j) \mbox{ and } uv \in E(G)\}$. 

Let $G$ be a graph. A {\em path decomposition} of a graph $G$ is a pair $\mathcal{T} = (P,\beta: V(P) \rightarrow 2^{V(G)})$, where 
$P$ is a path where every node $p\in V(P)$ 
is assigned a subset $\beta(p)\subseteq V(G)$, called a bag, such that 
the following conditions hold: (i) $\bigcup_{p\in V(P)}{\beta(p)}=V(G)$, (ii) for every edge $xy\in E(G)$ there is a $p\in V(P)$ such that  $\{x,y\}\subseteq \beta(p)$, and (iii) for any $v\in V(G)$ the subgraph of $P$ induced by the set  $\{p\mid v\in \beta(p)\}$ is connected. A path decomposition will also be denoted as a sequence of bags $\{\beta(p_1),\beta(p_2),\ldots,\beta(p_q)\}$ where $P = p_1p_2\ldots p_q$. The {\em width} of a path decomposition is $\max_{p\in V(P)} |\beta(p)| -1$. The {\em pathwidth} of $G$ is the  minimum width over all path decompositions of $G$ and is denoted by ${\sf pw}(G)$. Given a path decomposition of a graph $G$, we say it is rooted at exactly one of the two degree one vertices of the underlying path.

\begin{definition}~\cite{fomin2019finding}\label{def:cell-pathdecomp}
A path decomposition $\mathcal{T} = (P,\beta)$ of a clique-grid graph $G$ with representation $f:V(G) \rightarrow [p] \times [p']$ is a nice $\ell$-clique path decomposition ($\ell$-NCPD) if for the root $r$ of $P$, $\beta(r) = \emptyset$ and for each $v \in V(P)$ the following hold:

\begin{enumerate}
 \item There are at most $\ell$ cells $\{(i_1,j_1),(i_2,j_2),\ldots,(i_\ell,j_\ell)\}$ such that $\beta(v) = \bigcup_{p=1}^{\ell} f^{-1}(i_p,j_p)$,
 \item The node $v$ is one of the following types: (i) Leaf node where $\beta(v) = \emptyset$, (ii) Forget node where $v$ has exactly one child $u$ and there is a cell $(i,j) \in [p] \times [p']$ such that $f^{-1}(i,j) \subseteq \beta(u)$ and $\beta(v) = \beta(u) \setminus f^{-1}(i,j)$, (iii) Introduce node where $v$ has exactly one child $u$ and there is a cell $(i,j) \in [p] \times [p']$ such that $f^{-1}(i,j) \subseteq \beta(v)$ and $\beta(u) = \beta(v) \setminus f^{-1}(i,j)$, 
\end{enumerate}
\end{definition}

See Figure~\ref{l-ncpd} for an example of an NCPD. A path decomposition for a clique-grid graph $G$ with representation $f$ where only property $1$ of Definition~\ref{def:cell-pathdecomp} is true for a positive number $\ell$ is referred to as an $\ell$-CPD. 

\begin{figure}[h]
\centering
\includegraphics[scale=.8]{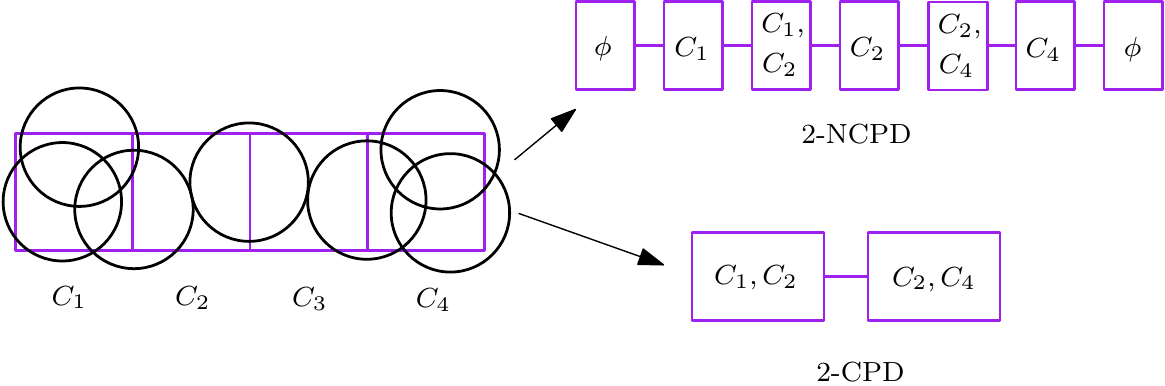}
\caption{An illustration of nice $2$-clique path decomposition.}
\label{l-ncpd}
\end{figure}

\section{NP-Hardness of \stree on Unit Disk Graphs}\label{sec:nph}

In this section, we consider the \stree problem on unit disk graphs and prove that this problem is \NPH. 
We show a reduction from {\sc Connected Vertex Cover} in planar graphs with maximum degree $4$. The reduction is very similar to that in~\cite{abu2015euclidean}.

\begin{theorem}
\label{nphard}
The \stree problem on unit disk graphs is \NPH.
\end{theorem}

\begin{proof}
We show a reduction from the {\sc Connected Vertex Cover} in planar graphs with maximum degree $4$ problem, which is known to be \NPH \cite{garey1977rectilinear}. 
Given a planar graph $G$ with maximum degree $4$ and an integer $k$, the {\sc Connected Vertex Cover} problem asks to find if there exists a vertex cover $D$ for $G$ such that the subgraph induced by $D$ is connected and $|D|\le k$. 
We adopt the proof of Abu-Affash \cite{abu2015euclidean}, where it was shown that the $k$-{\sc Bottleneck Full Steiner Tree} problem is \NPH. 
We make this reduction compatible for unit disk graphs.
Given a planar graph $G$ with maximum degree $4$ and an integer $k$,
we construct an unit disk graph $G_{\mathcal{C}}$ where $V(G_{\mathcal{C}})=\mathcal{C}$ in polynomial time, where $V(G_{\mathcal{C}})$ is divided into two sets of unit disks $R$ and $S$, denoted by Steiner and terminals, respectively.
Let $V(G)=\{v_1,v_2,\ldots,v_n\}$ and let $E(G)=\{e_1,e_2,\ldots,e_m\}$. Then, we compute an integer $k'$ such that $G$ has a connected vertex cover $D$ of size $k$ if and only if there exists a \stree with at most $k'$ Steiner vertices of $G_{\mathcal{C}}$. 

As as an intermediate step we build a rectangular grid graph $G'$.
First, we embed $G$ on a rectangular grid, with distance at least $8$ between adjacent vertices. Each vertex $v_i\in V(G)$ corresponds to a grid vertex, and each edge $e=v_iv_j\in E(G)$ corresponds to a rectilinear path comprised of some horizontal and vertical grid segments with endpoints corresponding to $v_i$ and $v_j$. Let $V(G')=\{v'_1,\ldots,v'_n\}$ be the grid points corresponding to the vertices of $V(G)$, and let $E(G')=\{p_{e_1},\ldots,p_{e_m}\}$ 
be the set of paths corresponding to the edges of $E(G)$
Moreover, these paths are pairwise disjoint; see Figure~\ref{np1}(b).
This embedding can be done in $O(n)$ time and the size of the grid is at most $n-2$ by $n-2$; see \cite{schnyder1990embedding}.
Next, we construct an unit disk graph $G_{\mathcal{C}}$ from $G'$.
First, we replace each grid vertex $v'_i\in V(G')$ by an unit disk. Let $C=\{c_1,\ldots,c_n\}$ be the set of unit disks centered at the grid points corresponding to the vertices of $V(G')$. 
For the sake of explanation we call these disks grid point disks. At this point, the unit disk graph is not connected due to the edge length which we have taken between any two adjacent vertices in the grid graph. In fact this length ensures that there are no undesirable paths other than the ones in $G$. Next, we place two sets of disks on each path $p_{e_i}\in E(G')$. Let $|p_{e_i}|$ be the total length of the grid segments of $p_{e_i}$. 
We place two Steiner disks on $p_{e_i}$, such that each one of them is adjacent to a grid point disk corresponding to $p_{e_i}$ and the distance between their centers is exactly $2$. Next, we place $|p_{e_i}|-6/2$ many terminals disks on $p_{e_i}$ such that the distance between any two adjacent centers is exactly $2$. 
See Figure~\ref{np1}(c) for detailed explanation. 
Let $s(e_i)$ be the set of Steiner disks and 
$t(e_i)$ be the set of terminal disks placed to $p_{e_i}$. The terminal set $R=\underset{e_i\in E(G')}{\bigcup} t(e_i)$; the Steiner set $S=C\cup \underset{e_i\in E(G')}{\bigcup} s(e_i)$.
$V(G_{\mathcal{C}})=R\cup S$ and $G_{\mathcal{C}}$ is the intersection graph induced by $V(G_{\mathcal{C}})$. Finally, we set $k'=m+2k-1$. 
Observe that, for any path $p_{e_i}$, the terminal set $t(e_i)$ itself form a Steiner tree without any Steiner disks. However, in order to make that tree connected we need at least one of Steiner disks from $s(e_i)$. This completes the construction.

\begin{figure}[h]
\centering
\includegraphics[scale=.8]{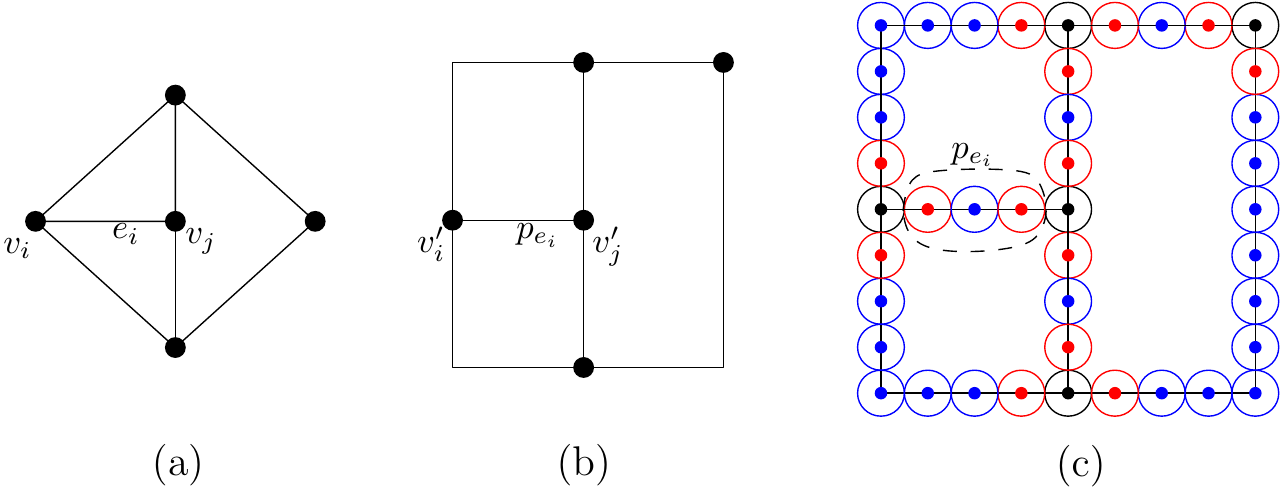}
\caption{(a) A planar graph $G$ of maximum degree $4$,
(b) the intermediate rectilinear embedding $G'$ of $G$,
(c) the unit disk graph $G_{\mathcal{C}}$; the black disks are corresponding to the grid vertices of $G'$, the blue disks are Steiner disks and the red disks are the terminal disks.}
\label{np1}
\end{figure}

In the forward direction, suppose $G$ has a connected vertex cover $D$ of size at most $k$. We construct a Steiner tree of $R$ in the following manner. 
For each edge $e_i$, we simply take the terminal path induced by $t(e_i)$. Now, let $T_S$ be any spanning tree of the subgraph of $G$ induced by $D$, containing $|D|-1$ edges.  The existence of such a spanning tree is ensured since $D$ is a connected vertex cover of $G$. For each edge $e=v_iv_j\in T_S$ we connect the corresponding disks $c_i,c_j$ by two Steiner red disks adjacent to them. 
Then, for each edge $e=v_iv_j\in G\setminus T_S$ we select one endpoint that is in $D$ (say $v_i$) and connect $c_i$ to the tree by its adjacent disk. 
The constructed tree is a Steiner tree of $R$ consisting $|D|+2(|D|-1)+(m-(|D|-1))$ which is $m+2k-1$.

Conversely, let there exists a Steiner tree $T$ of $R$ with at most $k'$ Steiner disks. Let $D\subseteq C$ be the set of vertices that appear in $T$,
and let $T'$ be the subtree of $T$ spanning over $D$. 
For each subset $t(e_i)\subseteq R$, let $T_{e_i}$ be the subtree of $T_{e_i}$ spanning the vertices in $t(e_i)$. By the above construction, $T_{e_i}$ does not require any Steiner disk. Moreover, it is easy to see that in any valid solution $T_{e_i}$ must be connected to at least one endpoint of $D$.
This implies that the set of vertices in $G$ corresponding to the vertices in $D$ is a connected vertex cover of $G$. 
Moreover a tree $T_{e_i}$ which also a subtree of $T$ is connected to $D$ via two Steiner disks of $s(e_i)$. Therefore, $T_S$ contains $|D|+2(|D|-1)+(m-(|D|-1))$ many Steiner disks. We started with the tree $T$ with at most $k'=m+2k-1$ many Steiner disks. This completes the proof.  
\end{proof}
\section{Subexponential Exact Algorithm for \stree on Unit Disk Graphs}\label{sec:subexp}
In this section, we prove Theorem~\ref{thm:subexp-UDG} by designing a sub-exponential algorithm for the \stree problem on unit disk graphs parameterized by $t+k$, where $t$ is the number of terminals and $k$ is an upper bound on the number of Steiner vertices.
In fact, our aim for this section is to design a subexponential algorithm for \stree on clique-grid graphs and as unit disk graphs are clique-grid graphs~\cite{fomin2019finding}, this would imply the algorithm proposed in Theorem~\ref{thm:subexp-UDG}.

\begin{lemma}\label{lem:subexp-CGG}
 The \stree problem on clique-grid graphs can be solved in $n^{O(\sqrt{t+k})}$ time.
\end{lemma}

For the rest of the section, we concentrate on proving Lemma~\ref{lem:subexp-CGG}. Informally, we first apply a Baker-like shifting strategy to create a family $\mathcal{F}$ of instances of \exst that preserves the answer for the input instance $(G,R,t,k)$ of \stree: if $(G,R,t,k)$ is a yes-instance then there is at least one constructed instance in $\mathcal{F}$ that is a yes-instance of \exst; if $(G,R,t,k)$ is a no-instance of \stree then all instances of $\mathcal{F}$ are no-instances of \exst. As a second step, we design a dynamic programming subroutine to solve \exst on each of the constructed instances of $\mathcal{F}$, which is enough to solve the \stree problem on $(G,R,t,k)$. 

Before we describe the subexponential algorithm, we state some properties of Steiner trees in clique-grid graphs.

\begin{observation}\label{obs:cell-cell-edge} 
 Consider a $k$-Steiner tree $T$ for a clique-grid graph $G$ with representation $f$, such that the set $\{uv\in E(T) \vert f(u) \neq f(v)\}$ is minimised over all $k$-Steiner trees for $G$. Let $\mathcal{C} = (i,j)$ be a cell of $G$. Then there are at most $\edgeconst$ edges with one endpoint in $\mathcal{C}$ and the other endpoint in another cell. 
\end{observation}

\begin{proof}
We claim that in the $k$-Steiner tree where the set $\{uv\in E(T) \vert f(u) \neq f(v)\}$ is minimised, there can be at most one neighbour of $\mathcal{C}$ in each cell $\mathcal{C}'\neq \mathcal{C}$. Suppose that $\mathcal{C}'$ is a cell that contains at least two neighbours of $\mathcal{C}$. Let two such neighbours be $u',v'$. Note that $u'v'$ is an edge in $E(G)$. Let $u,v$ (may be the same) be the neighbours of $u,v$, respectively in $\mathcal{C}$. Note that $uv$ is an edge in $E(G)$. Thus adding the edge $u'v'$ and removing the edge $uu'$ results in a connected graph containing all the terminals. The spanning tree of this connected graph has strictly less number of edges with endpoints in different cells, which is a contradiction to the choice of $T$.  

By the definition of clique-grid graphs, $\vert i-i'\vert , \vert j-j'\vert \leq 2$. Thus, when we fix a cell $\mathcal{C}$ there are at most $\cellconst$ cells that can have neighbours of vertices in $\mathcal{C}$. Putting everything together, for the $k$-Steiner tree $T$ where the set $\{uv\in E(T) \vert f(u) \neq f(v)\}$ is minimised, $\vert \{v \vert f(v) \neq (i,j), \exists u$ such that $f(u) = (i,j), uv \in E(G) \} \vert \leq \edgeconst$.
\end{proof}

\begin{observation}\label{obs:cell-bd} 
 Suppose there is a $k$-Steiner tree for a clique-grid graph $G$, and let $T$ be a $k$-Steiner tree where the set $\{uv\in E(T) \vert f(u) \neq f(v)\}$ is minimised. Moreover, amongst $k$-Steiner trees where $\{uv\in E(T) \vert f(u) \neq f(v)\}$ is minimised, $T$ has minimum number of Steiner points. Then, in $T$ the number of Steiner vertices per cell is at most $\vertconst$.
\end{observation}

\begin{proof}
 For the sake of contradiction, let $\mathcal{C}=(i,j)$ be a cell such that $\vert f^{-1}(i,j) \cap V(T) \vert \geq \vertconst +1$. Then by Observation~\ref{obs:cell-cell-edge}, there is at least one Steiner vertex $v \in f^{-1}(i,j) \cap V(T)$ such that it does not have any neighbours in $T \setminus f^{-1}(i,j)$. Consider the subgraph $T \setminus \{v\}$. Since the vertices of $f^{-1}(i,j)$ induce a clique, $T \setminus \{v\}$ is still a connected subgraph that contains all the terminals and strictly less number of Steiner vertices. Thus, a spanning tree of this connected subgraph contradicts the choice of $T$.   
\end{proof}

Consider a $k$-Steiner tree $T$ for an instance $(G,R,t,k)$ of \stree where $\{uv\in E(T) \vert f(u) \neq f(v)\}$ is minimised and then the number of Steiner vertices is minimised. By Observation~\ref{obs:exact}, $T$ is an exact $k'$-Steiner tree for the instance $(G,R,t,k')$ of \exst for some $k' \leq k$. Next, we define a {\em good family of instances} that preserve the answer for $(G,R,t,k)$ of \stree.

\begin{definition}\label{def:good-fam}
 For an instance $(G,R,t,k)$ of \stree on clique-grid graphs where $G$ has representation $f$, a good family of instances $\mathcal{F}$ has the following properties:
 \begin{enumerate}
  \item For each instance $(H,R,t,k')$ in the family, the input graph $H$ is an induced subgraph of $G$ that contains all vertices in $R$ and $k'\leq k$. Note that $H$ is also a clique-grid graph where $f\vert_{V(H)}$ is a representation.
  \item $(G,R,t,k)$ is a yes-instance of \stree if and only if there exists an instance $(H,R,t,k')\in \mathcal{F}$ which is a yes-instance of \exst.
  \item For any instance $(H,R,t,k') \in \mathcal{F}$, $H$ has a $\pwidth \sqrt{t+k}$-NCPD.
 \end{enumerate}
\end{definition}

We show that given an instance $(G,R,t,k)$ of \stree on clique-grid graphs, a good family of instances can be enumerated in subexponential time. 

\begin{lemma}\label{lem:compute-good-fam}
 Given an instance $(G,R,t,k)$ for \stree on clique-grid graphs  with $G$ represented by $f$, a good family of instances $\mathcal{F}$ can be computed in $n^{O(\sqrt{t+k})}$ time. 
\end{lemma}

\begin{proof}
 Let $T$ be a $k$-Steiner tree for $G$. In particular, $T$ is an exact $k'$-Steiner tree for some $k' \leq k$ and $V(T) = t+k' \leq t+k$. First, we employ a Baker-like technique similar to~\cite{fomin2019finding} (please refer to Figure~\ref{sub-exp-partition}). Note that if $G$ has $n$ vertices and has representation $f: V(G) \rightarrow [p]\times [p']$, then $p,p' \leq n$. Thus, $f$ represents $G$ on the $n \times n$ grid. First we define a column of the $n\times n$ grid. For any $j\in [n]$ the set of cells $\{(i, j) | i \in [n]\}$ is called a column. There are $n$ columns for the $n\times n$ grid. We partition the $n$ columns of the $n \times n$ grid with $n/2$ blocks of two consecutive columns and label them from the set of labels $[\sqrt{t+k}]$. Formally, each set of consecutive columns $\{2i-1,2i\}$, where $i \in [n/2]$ is labelled with $i \mbox{ mod } \sqrt{t+k}$. Thus, all the two consecutive columns $\{2i-1,2i\}$ are labelled with $i \mbox{ mod } \sqrt{t+k}$. 

Recall that an exact $k'$-Steiner tree $T$ has at most $t+k$ vertices. Applying the pigeonhole principle, there is a label $\ell \in \{1, 2, \ldots ,\sqrt{t+k}\}$ such that the number of vertices from $V(T)$ which are in columns labelled $\ell$ is at most $\sqrt{t+k}$. As we do not know this $k'$-Steiner tree $T$, we guess the Steiner vertices of $V(T)$ which are in the columns labelled $\ell$. The number of potential guesses is bounded by $n^{O(\sqrt{t+k})}$. Suppose $Y'$ is the set of guessed Steiner vertices of $V(T)$ which are in the columns labelled by $\ell$. Then we delete all the non-terminal vertices in columns labelled $\ell$, except the vertices of $Y'$ . Let $S$ be the set of deleted non-terminal vertices. Let $Y_{R}$ be the set of terminal vertices that are in columns labelled by $\ell$. Let $Y = Y' \cup Y_{R}$. Notice that by choice of label $\ell$, $\lvert Y\rvert \leq \sqrt{t+k}$. By Property $2$ of clique-grid graphs, $G\setminus (S\cup Y)$ is a disjoint union of clique-grid graphs each of which
is represented by a function with at most $2 \sqrt{t+k}$ columns. Formally, $G_1 = G[\bigcup_{j=1}^{2(\ell-1)} f^{-1}(*,j)]$ and $G_{i+1} = G[\bigcup_{j=i\cdot 2\ell +1}^{{\sf min}\{i \cdot2\ell +2\sqrt{t+k},n\}} f^{-1}(*,j)]$ for each $i \in \{1,\ldots,n/\sqrt{t+k}\}$. Each $G_i$ is a clique-grid graph with representation $f_i: V(G_i) \rightarrow [n]\times[2\sqrt{t+k}]$ defined as, $f_i(u) = (r, j)$, when $f(u) = (r, (i-1) 2\ell +j)$. Thus, by Property $2$ of Definition~\ref{def:clique-grid}, $G\setminus (S \cup Y) = G_1\uplus \ldots\uplus G_{n/\sqrt{t+k}}$.

\begin{figure}[h]
\centering
\includegraphics[scale=.8]{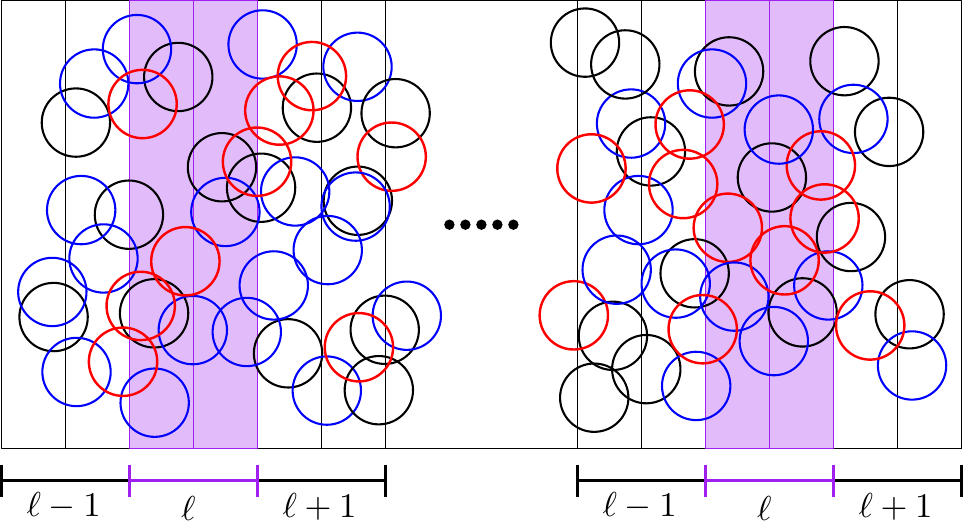}
\caption{An illustration of grid labelling. The blue disks are terminals, and the red and black disks are chosen Steiner vertices and not-chosen non-terminal vertices, respectively.}
\label{sub-exp-partition}
\end{figure}

\begin{claim}\label{clm:NCPD-graph}
 The graph $G \setminus S$ has a $\pwidth \sqrt{t+k}$-NCPD.
\end{claim}

\begin{proof}
Suppose we are able to show that for each $i \in \{1,\ldots, n/\sqrt{t+k}\}$ $G_i$ has a $\pwidthtwo \sqrt{t+k}$-CPD. This results in a $\pwidthtwo \sqrt{t+k}$-CPD for $G \setminus (S \cup Y) = G_1\uplus \ldots\uplus G_{n/\sqrt{t+k}}$. Finally, note that $\vert Y \vert \leq \sqrt{t+k}$ and therefore the vertices of $Y$ can belong to at most $\sqrt{t+k}$ cells. We add $Y$ to all the bags in the $\pwidthtwo\sqrt{t+k}$-CPD for $G \setminus (S \cup Y)$ to obtain a $\pwidth \sqrt{t+k}$-CPD for $G \setminus S$. We convert the $\pwidth\sqrt{t+k}$-CPD of $G \setminus S$ into a NCPD using the known algorithm of~\cite{Bodlaender96alinear}. Note that this results in a $\pwidth \sqrt{t+k}$-NCPD. 

What is left to show is that for each $G_i$ there is a $\pwidthtwo \sqrt{t+k}$-CPD. First, for each $G_i$, we give a path decomposition with the following sequence of bags: $\{X_1,X_2,\ldots, X_{n-2}\}$. This is done by defining each $X_i = f^{-1}(i,*) \cup f^{-1}(i+1,*) \cup f^{-1}(i+2,*)$. It is easy to check that this is a path decomposition of $G_i$. Note that since $G_i$ has at most $2\sqrt{t+k}$ columns, the number of cells contained in each $X_j, j \in [n-1]$ is at most $\pwidthtwo \sqrt{t+k}$. 
\end{proof}

Finally, notice that from the definition of the constructed instances keeping in mind potential $k$-Steiner trees, $(G,R,t,k)$ is a yes-instance of \stree if and only if there is an instance $(H,R,t,k') \in \mathcal{F}$ such that it is a yes-instance of \exst. Thus, accounting for guessing a label $\ell \in [\sqrt{t+k}]$ and the set $Y$ of Steiner vertices and terminal vertices of a potential solution Steiner tree that belong to columns labelled $\ell$, we obtain a good family of $n^{O(\sqrt{t+k})}$ instances for the given instance $(G,R,t,k)$.
\end{proof}

For the ease of our algorithm design, we make a slight modification of the NCPD for a constructed instance $(H,R,t,k') \in \mathcal{F}$: Upon fixing the label $\ell$ and a set $Y$ of terminal vertices and potential Steiner vertices in the columns labelled by $\ell$, we add the set $Y$ in all the bags of the resulting NCPD for $G \setminus S$. Therefore, no bag is empty after this modification. In particular the first and the last bags of the modified path decomposition contain only the set $Y$. Also notice that as $\vert Y \vert \leq \sqrt{t+k}$, the new path decomposition of $H$ is still an $O(\sqrt{t+k})$-CPD. We call this new path decomposition of $H$ a \emph{modified NCPD}. Now, we are ready to prove Lemma~\ref{lem:subexp-CGG}

\begin{proof}[Proof of Lemma~\ref{lem:subexp-CGG}]
As a first step of the algorithm, by Lemma~\ref{lem:compute-good-fam} in $n^{O(\sqrt{t+k})}$ time we compute a good family of instances $\mathcal{F}$ for the given instance $(G,R,t,k)$ of \stree on clique-grid graphs. From Definition~\ref{def:good-fam}(2), $(G,R,t,k)$ is a yes-instance of \stree if and only if there is an instance $(H,R,t,k') \in \mathcal{F}$ that is a yes-instance of \exst. Deriving from Definition~\ref{def:good-fam}(3), Lemma~\ref{lem:compute-good-fam} and the construction of a modified NCPD, for each instance $(H,R,t,k') \in \mathcal{F}$, there is a modified $O(\sqrt{t+k})$-NCPD for $H$, due to a guessed label $\ell$ and a guessed set $Y$ of non-terminal vertices  from columns labelled by $\ell$ such that the following hold: (i) $\vert Y \vert \leq \sqrt{t+k}$, (ii) if $(H,R,t,k')$ is a yes-instance then there is an exact $k'$-Steiner tree $T$ such that all vertices of $Y$ are Steiner vertices in $T$. Let the modified NCPD using the set $Y$ have the sequence of bags $\{X_1,X_2,\ldots, X_q\}$. Recall that the definition of the modified NCPD ensures that $X_1 = X_q = Y$.

In the next step, our algorithm for \stree considers every instance $(H,R,t,k') \in \mathcal{F}$ and checks if it is a yes-instance of \exst. By Definition~\ref{def:good-fam}(2), this is sufficient to determine if $(G,R,t,k)$ is a yes-instance of \stree. 

For the rest of the proof we design a dynamic programming subroutine algorithm $\mathcal{A}$ for \exst that takes as input an instance $(H,R,t,k') \in \mathcal{F}$ and uses its modified $O(\sqrt{t+k})$-NCPD to determine whether it is a yes-instance of \exst. Suppose $(G,R,t,k)$ is a yes-instance and consider a $k$-Steiner tree $T$ for $(G,R,t,k)$ where $\{uv\in E(T) \vert f(u) \neq f(v)\}$ is minimised and then the number of Steiner vertices in $T$ is minimised. Using Observation~\ref{obs:exact}, this is an exact $k'$-Steiner tree of $G$ for some $k' \leq k$. By the construction in Lemma~\ref{lem:compute-good-fam} note that there is an instance $(H,R,t,k') \in \mathcal{F}$ such that $T$ is an exact $k'$-Steiner tree for $(H,R,t,k')$. The aim of the dynamic programming algorithm is to correctly determine that this particular instance $(H,R,t,k')$ is a yes-instance. The algorithm $\mathcal{A}$ is designed in such a manner that for such a yes-instance $(H,R,t,k')$ the tree $T$ will be the potential solution Steiner tree that behaves as a certificate of correctness.

The states of the dynamic programming algorithm store information required to represent the partial solution Steiner tree, which is the potential solution Steiner tree restricted to the graph seen so far. The states are of the form $\mathcal{A}$[$\ell, Q, \mathcal{Q} =Q_1\uplus Q_2 \ldots \uplus Q_b, \mathcal{P} = P_1\uplus\ldots P_b,k''$] where: 
\begin{itemize}
\item $\ell \in [q]$ denotes the index of the bag $X_\ell$ of the modified NCPD of $H$.
\item $Q \subseteq X_\ell \setminus R$ is a set of at most $\vertconst \cdot \pwidth$ non-terminal vertices. For each cell $\mathcal{C} = (i,j)$ that belongs to $X_\ell$, $\vert Q \cap f^{-1}(i,j) \vert \leq \vertconst $. 
\item $\mathcal{Q} = Q_1\uplus Q_2 \ldots \uplus Q_b$ is a partition of $Q$ with the property that for each cell $\mathcal{C} = (i,j)$,  $Q \cap f^{-1}(i,j)$ is contained completely in exactly one part of $\mathcal{Q}$. 
 \item The partition $\mathcal{P}$ is over the vertex set $Q \cup (R \cap X_\ell)$. $Q \cap P_i = Q_i$. Also for each cell $\mathcal{C}$ in $X_\ell$, $\mathcal{C} \cap (Q \cup R)$ is completely contained in exactly one part of $\mathcal{P}$. 
 \item The value $k''$ represents the total number of Steiner vertices used so far in this partial solution Steiner tree. $\vert Q \vert \leq k''$ holds. 
\end{itemize}

Essentially, let $T$ be an exact $k'$-Steiner tree for $(H,R,t,k')$ if it is a yes-instance. For $\ell \in [q]$, let $T_{\sf ptl}^\ell$ represent the partial solution Steiner tree when $T$ is restricted to $H[\bigcup_{j=1}^{\ell} X_j]$. The partition $\mathcal{P}$ represents the intersection of a component of $T_{\sc ptl}^{\ell}$  with $X_\ell$. The set $Q$ is the set of Steiner vertices of $T_{\sf ptl}^{\ell}$ in the bag $X_\ell$ and $\mathcal{Q}$ is the partition of $Q$ with respect to the components of $T_{\sc ptl}^{\ell}$. The number $k''$ denotes the total number of Steiner vertices in $T_{\sf ptl}^\ell$.

In order to show the correctness of $\mathcal{A}$ we need to maintain the following invariant throughout the algorithm:
(LHS) $\mathcal{A}[\ell,Q,\mathcal{Q}=Q_1 \uplus Q_2,\ldots Q_b,\mathcal{P} = P_1\uplus P_2 \uplus P_b,k']=1$ if and only if (RHS) there is a forest $T'$ as a subgraph of $H[\bigcup_{j=1}^{\ell}]$ with $b$ connected components $D_1,\ldots,D_b$: $D_i \cap X_\ell = P_i$, $(D_i \setminus R)  \cap X_\ell = Q_i$, the total number of non-terminal points in $T'$ is $k''$, for each cell $\mathcal{C}$ the number of nonterminal vertices in $\mathcal{C} \cap T'$ is at most $\vertconst$, and $R \cap (\bigcup_{j=1}^{\ell} X_j) \subseteq V(T')$.

Suppose the algorithm invariant is true. This means that if $\mathcal{A}[q,Y, {Y},{Y},k']=1$ then there is an exact $k'$-Steiner tree for $(H,R,t,k')$. On the other hand, suppose $(G,R,t,k)$ is a yes-instance and has a $k$-Steiner tree $T$ where $\{uv\in E(T) \vert f(u) \neq f(v)\}$ is minimised and then the number of Steiner vertices in $T$ is minimised. By Observation~\ref{obs:cell-bd}, the number of Steiner vertices of $T$ in each cell of $G$ is bounded by $\vertconst$. By Observation~\ref{obs:exact} and the construction in Lemma~\ref{lem:compute-good-fam} note that there is a subset $Y$ and an instance $(H,R,t,k') \in mathcal{F}$ such that $T$ is an exact $k'$-Steiner tree for $(H,R,t,k')$ and $Y \subseteq V(T)$. Suppose the invariant of the algorithm is true. This means that if $(G,R,t,k)$ is a yes-instance of \stree then there is a $(H,R,t,k')$ for which $\mathcal{A}[q,Y,{Y},{Y},k']=1$.

Thus, proving the correctness of the algorithm $\mathcal{A}$ amounts to proving the correctness of the invariant of $\mathcal{A}$. We prove the correctness of the invariant by induction on $\ell$.
If $\ell =1$ then $X_\ell$ must be a {\bf leaf bag}. By definition of the modified NCPD, the bag contains $Y$. 

$\mathcal{A}[1,Q,\mathcal{Q},\mathcal{P},k''] = 1$ if $Q = Y$, $\mathcal{Q}$ is the partition of $Y$ into the connected components in $H[Y]$, $\mathcal{P} = \mathcal{Q}$, $k'' = \vert Y \vert$. In all other cases, $\mathcal{A}[1,Q,\mathcal{Q},\mathcal{P},k''] = 0$. 

First, suppose $\mathcal{A}[1,Q,\mathcal{Q},\mathcal{P},k''] = 1$. Then as $X_1$ does not contain any terminal vertices, (RHS) trivially is true for the cases when $\mathcal{A}[1,Q,\mathcal{Q},\mathcal{P},k''] = 1$. On the other hand, suppose (RHS) is true for $\ell = 1$. Again considering the cases when $\mathcal{A}[1,Q,\mathcal{Q},\mathcal{P},k''] = 1$, (LHS) holds. So the invariant holds when $\ell =1$.

Now, we assume that $\ell >1$. Our induction hypothesis is that the invariant of the algorithm is true for all $1 \leq \ell' <\ell$. We show that the invariant is true for $\ell$. There can be two cases:

\vspace{-.2cm}

\subparagraph*{Case 1:}
$X_\ell$ is a {\bf forget bag} with exactly one child $X_{\ell-1}$ : Let $\mathcal{C}$ be the cell being forgotten in $X_\ell$. Consider $\mathcal{A}[\ell,Q, \mathcal{Q}=Q_1,\ldots Q_b,\mathcal{P}=P_1\ldots P_b,k'']$.

Let $Q' \subseteq X_{\ell-1} \setminus R$ such that $Q \subseteq Q'$ and $Q' \setminus Q$ consists of a set of at most $\vertconst$ non-terminal vertices from $\mathcal{C}$. Let $\mathcal{P}'=P'_1 \ldots P'_b$ be a partition of $(Q' \cup R) \cap X_{\ell-1})$ such for each cell $\mathcal{C}'$ in $X_{\ell-1}$, $\mathcal{C}' \cap (Q' \cup R)$ is completely contained in exactly one part. Also, $P_i = P_i' \setminus \mathcal{C}$. Moreover, consider the part $P'_i$ such that $\mathcal{C} \cap (Q' \cup R) \subseteq P'_i$: $P'_i \setminus (\mathcal{C} \cap (Q' \cup R)) \neq \emptyset$. Let $\mathcal{Q}'$ be the partition of $Q'$ such that $Q' \cap P'_i = Q'_i$. 
If $\mathcal{A}[\ell-1,Q',\mathcal{Q}',\mathcal{P}',k''] = 1$ then $\mathcal{A}[\ell,Q,\mathcal{Q},\mathcal{P},k'']=1$. Otherwise, $\mathcal{A}[\ell,Q,\mathcal{P},k'']=0$.

Suppose (LHS) of the invariant is true for $\mathcal{A}[\ell,Q,\mathcal{Q},\mathcal{P},k'']$: $\mathcal{A}[\ell,Q,\mathcal{Q},\mathcal{P},k'']=1$.  By definition, there is a $\mathcal{A}[\ell-1,Q',\mathcal{Q}',\mathcal{P}',k''] = 1$ for a $Q',\mathcal{Q}',\mathcal{P}'$ as described above. By induction hypothesis, (RHS) corresponding to $\mathcal{A}[\ell-1,Q',\mathcal{Q}',\mathcal{P}',k''] = 1$ holds. Thus, there is a witness forest $T'$ in $H[\bigcup_{j=1}^{\ell-1} X_j] = H[\bigcup_{j=1}^{\ell}]$ (By definition of a forget bag). By definition of $Q,\mathcal{Q},\mathcal{P}$, $T'$ is also a witness forest in $H[\bigcup_{j=1}^{\ell} X_j]$ and therefore (RHS) is true for $\mathcal{A}[\ell,Q,\mathcal{Q},\mathcal{P},k'']$.

On the other hand, suppose (RHS) is true for $\mathcal{A}[\ell,Q,\mathcal{Q},\mathcal{P},k'']$. Then there is a witness forest $T'$ in $H[\bigcup_{j=1}^{\ell} X_j] = H[\bigcup_{j=1}^{\ell-1}]$. Moreover, $T'$ has $b$ connected components $D_1,\ldots,D_b$: $D_i \cap X_\ell = P_i$, $(D_i \setminus R)  \cap X_\ell = Q_i$, the total number of non-terminal points in $T'$ is $k''$ and $R \cap (\bigcup_{j=1}^{\ell} X_j) \subseteq V(T')$. Let $D_i \cap X_{\ell -1} =P'_i$, $(D_i \setminus R)  \cap X_{\ell-1} = Q'_i$, $Q' = \bigcup_{j=1}^{b} Q'_i$. Note that the total number of non-terminal points in $T'$ is $k''$ and by definition of a forget node it is still true that $R \cap (\bigcup_{j=1}^{\ell-1} X_j) \subseteq V(T')$. By induction hypothesis, (LHS) is true for $\mathcal{A}[\ell-1,Q',\mathcal{Q}',\mathcal{P}',k'']$ and $\mathcal{A}[\ell-1,Q',\mathcal{Q}',\mathcal{P}',k'']=1$. By the description above, this implies that $\mathcal{A}[\ell,Q,\mathcal{Q},\mathcal{P},k'']=1$. Therefore, (LHS) is true for $\mathcal{A}[\ell,Q,\mathcal{Q},\mathcal{P},k'']$. 

\vspace{-.2cm}
\subparagraph*{Case 2:}
$X_\ell$ is an {\bf introduce bag} with exactly one child $X_{\ell -1}$. Let $\mathcal{C}$ be the cell being introduced in $X_\ell$. Consider $\mathcal{A}[\ell,Q, \mathcal{Q}=Q_1,\ldots Q_b,\mathcal{P}=P_1\ldots P_b,k'']$. Without loss of generality, let $P_b$ contain all the vertices in $\mathcal{C} \cap (Q \cup R)$.

By definition of a state, $\vert \mathcal{C} \cap Q \vert \leq \vertconst$. Let ${\sf St} = \mathcal{C} \cap Q$ and $Q' = Q \setminus {\sf St}$. 
Let $\mathcal{P}' = P'_1\uplus P'_2\ldots \uplus P'_b \uplus \ldots P'_{d}$ be a partition of $Q' \cup (R \cap X_{\ell-1})$ such that for $j < b, P_j = P'_j$, and $P_b = \mathcal{C} \cap (Q \cup R) \cup \bigcup_{j=b}^{d} P'_j$. Moreover, $\mathcal{C} \cap (Q \cup R)$ has a neighbour in each $P'_j, b \leq j \leq d$. 
Let $\mathcal{Q}'$ be the partition of $Q'$ such that $Q' \cap P'_i = Q'_i$. 
Let $k^* = k'' - \vert {\sf St} \vert$.
If $\mathcal{A}[\ell-1,Q',\mathcal{Q}',\mathcal{P}',k^*] = 1$ then $\mathcal{A}[\ell,Q,\mathcal{Q},\mathcal{P},k'']=1$. Otherwise, $\mathcal{A}[\ell,Q,\mathcal{P},k'']=0$.

Suppose (LHS) of the invariant is true for $\mathcal{A}[\ell,Q,\mathcal{Q},\mathcal{P},k'']$: $\mathcal{A}[\ell,Q,\mathcal{Q},\mathcal{P},k'']=1$.  By definition, there is a $\mathcal{A}[\ell-1,Q',\mathcal{Q}',\mathcal{P}',k^*] = 1$ for a $Q',\mathcal{Q}',\mathcal{P}'$ as described above. By induction hypothesis, (RHS) corresponding to $\mathcal{A}[\ell-1,Q',\mathcal{Q}',\mathcal{P}',k^*] = 1$ holds. Thus, there is a witness forest $T'$ in $H[\bigcup_{j=1}^{\ell-1} X_j]$. By definition of $Q,\mathcal{Q},\mathcal{P}$, $H[V(T') \cup (\mathcal{C} \cap (Q \cup R))]$ is a connected graph. Consider a spanning tree of this connected graph. By definition of $k^*$, this spanning tree has all vertices of $R$ and exactly $k''$ non-terminal vertices. Therefore, this spanning tree is a witness forest in $H[\bigcup_{j=1}^{\ell} X_j]$ and therefore (RHS) is true for $\mathcal{A}[\ell,Q,\mathcal{Q},\mathcal{P},k'']$.

On the other hand, suppose (RHS) is true for $\mathcal{A}[\ell,Q,\mathcal{Q},\mathcal{P},k'']$. Then there is a witness forest $T'$ in $H[\bigcup_{j=1}^{\ell} X_j]$. Moreover, $T'$ has $b$ connected components $D_1,\ldots,D_b$: $D_i \cap X_\ell = P_i$, $(D_i \setminus R)  \cap X_\ell = Q_i$, the total number of non-terminal points in $T'$ is $k''$ and $R \cap (\bigcup_{j=1}^{\ell} X_j) \subseteq V(T')$. Without loss of generality, let $D_b$ contain $T' \cap \mathcal{C}$. Let $D'_1,D'_2,\ldots D'_b,\ldots, D'_d$ be the connected components of $T'$ restricted to $H[\bigcup_{j=1}^{\ell-1} X_j]$. Let $D'_i \cap X_{\ell -1} =P'_i$, $(D'_i \setminus R)  \cap X_{\ell-1} = Q'_i$, $Q' = \bigcup_{j=1}^{d} Q'_i$. Note that the total number of non-terminal points in $T'$ is $k^* = k'' - \vert {\sf St} \vert $ and by definition of an introduce node it is true that $R \cap (\bigcup_{j=1}^{\ell-1} X_j) \subseteq V(T') \cap (\bigcup_{j=1}^{\ell-1} X_j)$. By induction hypothesis, (LHS) is true for $\mathcal{A}[\ell-1,Q',\mathcal{Q}',\mathcal{P}',k^*]$ and $\mathcal{A}[\ell-1,Q',\mathcal{Q}',\mathcal{P}',k^*]=1$. By the description above, this implies that $\mathcal{A}[\ell,Q,\mathcal{Q},\mathcal{P},k'']=1$. Therefore, (LHS) is true for $\mathcal{A}[\ell,Q,\mathcal{Q},\mathcal{P},k'']$. 

Finally, we analyse the time complexity of the algorithm. First, the good family $\mathcal{F}$ is computed in $n^{O(\sqrt{t+k})}$ time as per Lemma~\ref{lem:compute-good-fam}, and the number of instances in the good family $\mathcal{F}$ is $n^{O(\sqrt{t+k})}$. For one such instance $(H,R,t,k')$ the possible states for the algorithm $\mathcal{A}$ are of the form $[\ell, Q,\mathcal{Q},\mathcal{P},k'']$. By definition, $\ell \leq n$, $k'' \leq k'$ and $Q = O(\sqrt{t+k})$. Again, by definition $\mathcal{P}$ is upper bounded by the number of partitions of cells contained in a bag of the modified NCPD of $(H,R,t,k')$. Thus, the number of possibilities of $\mathcal{P}$ is $\sqrt{t+k}^{O(\sqrt{t+k})})$. Also by definition, $\mathcal{Q}$ is fixed once $Q$ and $\mathcal{P}$ are fixed. Therefore, the number of possible states is $n^{O(\sqrt{t+k})}$. From the description of $\mathcal{A}$, the computation of $\mathcal{A}[\ell,Q,\mathcal{Q},\mathcal{P},k'']$ may look up the solution for $n^{O(\sqrt{t+k})}$ instances of the form $\mathcal{A}[\ell-1,Q',\mathcal{Q}',\mathcal{P}',k^*]$ and therefore takes $n^{O(\sqrt{t+k})}$ time. Thus, the total time for the dynamic programming is $O(n^{\sqrt{t+k}})$. 
\end{proof}

\old{
\begin{proof}[Proof Sketch]
As a first step of the algorithm, by Lemma~\ref{lem:compute-good-fam} in $n^{O(\sqrt{t+k})}$ time we compute a good family of instances $\mathcal{F}$ for the given instance $(G,R,t,k)$ of \stree on clique-grid graphs. From Definition~\ref{def:good-fam}(2), $(G,R,t,k)$ is a yes-instance of \stree if and only if there is an instance $(H,R,t,k') \in \mathcal{F}$ that is a yes-instance of \exst. Deriving from Definition~\ref{def:good-fam}(3), Lemma~\ref{lem:compute-good-fam} and the construction of a modified NCPD, for each instance $(H,R,t,k') \in \mathcal{F}$, there is a modified $O(\sqrt{t+k})$-NCPD for $H$, due to a guessed label $\ell$ and a guessed set $Y$ of non-terminal vertices  from columns labelled by $\ell$ such that the following hold: (i) $\vert Y \vert \leq \sqrt{t+k}$, (ii) if $(H,R,t,k')$ is a yes-instance then there is an exact $k'$-Steiner tree $T$ such that all vertices of $Y$ are Steiner vertices in $T$. Let the modified NCPD using the set $Y$ have the sequence of bags $\{X_1,X_2,\ldots, X_q\}$. Recall that the definition of the modified NCPD ensures that $X_1 = X_q = Y$.

In the next step, our algorithm for \stree considers every instance $(H,R,t,k') \in \mathcal{F}$ and checks if it is a yes-instance of \exst. By Definition~\ref{def:good-fam}(2), this is sufficient to determine if $(G,R,t,k)$ is a yes-instance of \stree. 

For the rest of the proof we design a dynamic programming subroutine algorithm $\mathcal{A}$ for \exst that takes as input an instance $(H,R,t,k') \in \mathcal{F}$ and uses its modified $O(\sqrt{t+k})$-NCPD to determine whether it is a yes-instance of \exst. Suppose $(G,R,t,k)$ is a yes-instance and consider a $k$-Steiner tree $T$ for $(G,R,t,k)$ where $\{uv\in E(T) \vert f(u) \neq f(v)\}$ is minimised and then the number of Steiner vertices in $T$ is minimised. Using Observation~\ref{obs:exact}, this is an exact $k'$-Steiner tree of $G$ for some $k' \leq k$. By the construction in Lemma~\ref{lem:compute-good-fam} note that there is an instance $(H,R,t,k') \in \mathcal{F}$ such that $T$ is an exact $k'$-Steiner tree for $(H,R,t,k')$. The aim of the dynamic programming algorithm is to correctly determine that this particular instance $(H,R,t,k')$ is a yes-instance. The algorithm $\mathcal{A}$ is designed in such a manner that for such a yes-instance $(H,R,t,k')$ the tree $T$ will be the potential solution Steiner tree that behaves as a certificate of correctness.

The states of the dynamic programming algorithm store information required to represent the partial solution Steiner tree, which is the potential solution Steiner tree restricted to the graph seen so far. The states are of the form $\mathcal{A}$[$\ell, Q, \mathcal{Q} =Q_1\uplus Q_2 \ldots \uplus Q_b, \mathcal{P} = P_1\uplus\ldots P_b,k''$] where: 
\begin{itemize}
\item $\ell \in [q]$ denotes the index of the bag $X_\ell$ of the modified NCPD of $H$.
\item $Q \subseteq X_\ell \setminus R$ is a set of at most $\vertconst \cdot \pwidth$ non-terminal vertices. For each cell $\mathcal{C} = (i,j)$ that belongs to $X_\ell$, $\vert Q \cap f^{-1}(i,j) \vert \leq \vertconst $. 
\item $\mathcal{Q} = Q_1\uplus Q_2 \ldots \uplus Q_b$ is a partition of $Q$ with the property that for each cell $\mathcal{C} = (i,j)$,  $Q \cap f^{-1}(i,j)$ is contained completely in exactly one part of $\mathcal{Q}$. 
 \item The partition $\mathcal{P}$ is over the vertex set $Q \cup (R \cap X_\ell)$. $Q \cap P_i = Q_i$. Also for each cell $\mathcal{C}$ in $X_\ell$, $\mathcal{C} \cap (Q \cup R)$ is completely contained in exactly one part of $\mathcal{P}$. 
 \item The value $k''$ represents the total number of Steiner vertices used so far in this partial solution Steiner tree. $\vert Q \vert \leq k''$ holds. 
\end{itemize}

Essentially, let $T$ be an exact $k'$-Steiner tree for $(H,R,t,k')$ if it is a yes-instance. For $\ell \in [q]$, let $T_{\sf ptl}^\ell$ represent the partial solution Steiner tree when $T$ is restricted to $H[\bigcup_{j=1}^{\ell} X_j]$. The partition $\mathcal{P}$ represents the intersection of a component of $T_{\sf ptl}^{\ell}$  with $X_\ell$. The set $Q$ is the set of Steiner vertices of $T_{\sf ptl}^{\ell}$ in the bag $X_\ell$ and $\mathcal{Q}$ is the partition of $Q$ with respect to the components of $T_{\sf ptl}^{\ell}$. The number $k''$ denotes the total number of Steiner vertices in $T_{\sf ptl}^\ell$.

In order to show the correctness of $\mathcal{A}$ we need to maintain the following invariant throughout the algorithm:
(LHS) $\mathcal{A}[\ell,Q,\mathcal{Q}=Q_1 \uplus Q_2,\ldots Q_b,\mathcal{P} = P_1\uplus P_2 \uplus P_b,k']=1$ if and only if (RHS) there is a forest $T'$ as a subgraph of $H[\bigcup_{j=1}^{\ell}]$ with $b$ connected components $D_1,\ldots,D_b$: $D_i \cap X_\ell = P_i$, $(D_i \setminus R)  \cap X_\ell = Q_i$, the total number of non-terminal points in $T'$ is $k''$, for each cell $\mathcal{C}$ the number of nonterminal vertices in $\mathcal{C} \cap T'$ is at most $\vertconst$, and $R \cap (\bigcup_{j=1}^{\ell} X_j) \subseteq V(T')$.

Suppose the algorithm invariant is true. This means that if $\mathcal{A}[q,Y, {Y},{Y},k']=1$ then there is an exact $k'$-Steiner tree for $(H,R,t,k')$. On the other hand, suppose $(G,R,t,k)$ is a yes-instance and has a $k$-Steiner tree $T$ where $\{uv\in E(T) \vert f(u) \neq f(v)\}$ is minimised and then the number of Steiner vertices in $T$ is minimised. By Observation~\ref{obs:cell-bd}, the number of Steiner vertices of $T$ in each cell of $G$ is bounded by $\vertconst$. By Observation~\ref{obs:exact} and the construction in Lemma~\ref{lem:compute-good-fam} note that there is a subset $Y$ and an instance $(H,R,t,k') \in \mathcal{F}$ such that $T$ is an exact $k'$-Steiner tree for $(H,R,t,k')$ and $Y \subseteq V(T)$. Suppose the invariant of the algorithm is true. This means that if $(G,R,t,k)$ is a yes-instance of \stree then there is a $(H,R,t,k')$ for which $\mathcal{A}[q,Y,{Y},{Y},k']=1$.

Thus, proving the correctness of the algorithm $\mathcal{A}$ amounts to proving the correctness of the invariant of $\mathcal{A}$. We prove the correctness of the invariant by induction on $\ell$. Due to paucity of space we defer the full proof including the correctness of the algorithm invariant and the running time analysis to Appendix~\ref{secapp:subexp}.
\end{proof}
}

\section{FPT Algorithm for \stree on Unit Disk Graphs}\label{sec:FPT}

In this section, we prove Theorem~\ref{thm:FPT-UDG}. We consider the \stree problem on unit disk graphs and design an FPT algorithm parameterized by $k$, which is an upper bound on the number of Steiner vertices in the solution Steiner tree. Our algorithm is based on the idea that for an instance $(G,R,t,k)$, in order to determine the existence of a Steiner tree we can first find spanning trees for all components of $G[R]$ and extend these spanning trees to a required $k$-Steiner tree.  

In fact, we prove our results for the superclass of clique-grid graphs. For an instance $(G,R,t,k)$ of \stree on clique-grid graphs, where $G$ has $n$ vertices and $R\subseteq V(G)$ is the set of terminals we prove the following result in this rest of this section. 

\begin{lemma}\label{lem:FPT}
\stree on clique-grid graphs has an FPT algorithm with running time $2^{O(k)}n^{O(1)}$.
\end{lemma}

First, we prove some properties of Steiner trees for unit disk graphs. Consider the induced subgraph $G[R]$. Let $C_1,C_2,\ldots,C_q$ be the connected components in $G[R]$. For each $C_i$, $i \in [q]$, let $T_i$ be a spanning tree of $C_i$.

\begin{observation}\label{obs:terminal-comps}
 Let $G$ be a clique-grid graph with the terminal set $R$. Let $C_1,C_2,\ldots,C_q$ be the connected components of $G[R]$, and for each $i \in [q]$ let $T_i$ be a spanning tree for each $C_i$. For any $k$, let $T'$ be a $k$-Steiner tree for $G$. Then there is a $k$-Steiner tree $T$ such that for each $i \in [q]$ $T_i$ is a subtree of $T$. Moreover, $q \leq \compconst k$.
\end{observation}

\begin{proof}
 Consider the $k$-Steiner tree $T$ and let $S = V(T) \setminus R$ be the set of Steiner vertices of $T$. Note that in $G[R \cup S]$, $T$ is a spanning tree and therefore $G[R \cup S]$ is a connected graph. Similarly, for each $i \in [q]$, $T_i$ is a subgraph of $G[R \cup S]$. Consider the subgraph $H = T' \cup \bigcup_{i\in[q]} T_i$. As $T'$ is a spanning tree, $T' \cup \bigcup_{i\in[q]} T_i$ is a connected graph. We consider an arbitrary ordering $\mathcal{O}$ of the edges in $E(H) \setminus (\bigcup_{i\in [q]} E(T_i))$. In this order we iteratively throw away an edge $e_j \in E(H) \setminus (\bigcup_{i\in [q]} E(T_i))$ if the resulting graph remains connected upon throwing $e_j$ away. Let $H'$ be the graph at the end of considering all the edges in the order $\mathcal{O}$. We prove that $H'$ must be a tree. Suppose for the sake of contradiction, there is a cycle $C$ as a subgraph of $H'$. As for each $i\in [q]$, $T_i$ is a tree and for each $i \neq i' \in [q]$, $V(T_i) \cap V(T_{i'}) = \emptyset$, there must be an edge from $E(H) \setminus (\bigcup_{i\in [q]} E(T_i))$ in $E(C)$. Consider the edge $e \in (E(H) \setminus (\bigcup_{i\in [q]} E(T_i))) \cap E(C)$ with the largest index according to $\mathcal{O}$. This edge was throwable as $C \setminus \{e\}$ ensured any connectivity due to $e$. Thus, there can be no cycle in $H'$ and it is a spanning tree of $V(H)$. This implies that $T=H'$ is a $k$-Steiner tree for $G$, $S$ being the set of at most $k$ Steiner vertices, such that for each $i \in [q]$, $T_i$ is a subtree of $T$. 

 Finally, we show that if a $k$-Steiner tree $T$ exists then $q \leq \compconst k$. Let $f$ be a representation of the clique-grid graph $G$. Note that for any cell $(a,b)$ $f^{-1}(a,b)$ is a clique, Therefore, there can be at most one component $C_i$ intersecting with a cell $(a,b)$. By property $(2)$ of Definition~\ref{def:clique-grid}, there are at most $\cellconst$ cells that can have neighbours of any vertex in $(a,b)$. Thus, for any Steiner vertex, there can be at most $\compconst$ components of $G[R]$ it can have neighbours in. Putting everything together, if there are at most $k$ Steiner vertices that are used to connect the $q$ connected components of $G[R]$ and each Steiner vertex can have neighbours in at most $\compconst$ components, then it must be that $q \leq \compconst k$.
  \end{proof}

Henceforth, we wish to find a solution $k$-Steiner tree $T$ such that for each $i \in [q]$, $T_i$ is a subtree of $T$. 

\begin{definition}\label{def:contract}
  Let $G$ be a clique-grid graph with the terminal set $R$. Let $C_1,C_2,\ldots,C_q$ be the connected components of $G[R]$, and for each $i \in [q]$ let $T_i$ be a spanning tree for each $C_i$. Let $G^*$ be the following graph: $V(G^*) = V(G \setminus R) \cup R^*$ where $R^* = \{c_i \vert i \in [q]\}$, $E(G^*) = \{v_1v_2 \vert v_1,v_2 \in V(G) \setminus R\} \cup \{vc_i \vert v \in V(G) \setminus R, \exists u \in C_i \mbox{ s.t } vu\in E(G) \}$. $G^*$ is called the component contracted graph of $G$ and $\{c_i\vert i \in [q]\}$ is the set of terminals for $G^*$ (See Figure~\ref{FPT-contraction}).
\end{definition}

\begin{figure}[h]
\centering
\includegraphics[scale=.8]{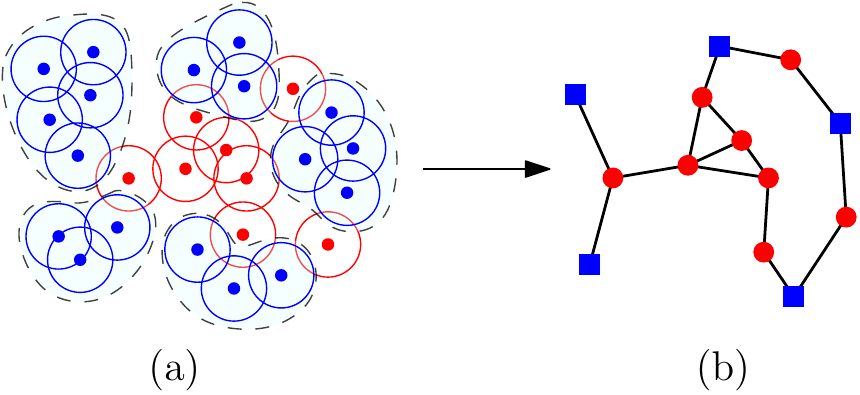}
\caption{An illustration of the component contraction; (a) red disks are Steiners and blue disks are terminals; (b) red vertices are Steiner vertices and blue vertices are contracted terminal components.}
\label{FPT-contraction}
\end{figure}

Note that $G^*$ may no longer be a clique-grid graph. From the definition of a component contracted graph and Observation~\ref{obs:terminal-comps}, we have the following observation.

\begin{observation}\label{obs:contract}
 Let $G$ be a clique-grid graph with the terminal set $R$. Let $C_1,C_2,\ldots,C_q$ be the connected components of $G[R]$, and for each $i \in [q]$ let $T_i$ be a spanning tree for each $C_i$. Let $G^*$ be the component contracted graph of $G$ using the $T_i$'s. Then $(G,R,t,k)$ is a yes-instance of \stree if and only if $q \leq \compconst k$ and $(G^*,R^*,q,k)$ is a yes-instance of \stree.
\end{observation}

Now we are ready to design our FPT algorithm for \stree on clique-grid graphs parameterized by $k$ and complete the proof of Lemma~\ref{lem:FPT}.

 \begin{proof}[Proof of Lemma~\ref{lem:FPT}]
 Let $(G,R,t,k)$ be an input instance of 
 $n$-vertex clique-grid graphs. Let $C_1,C_2,\ldots,C_q$ be the connected components of $G[R]$, and for each $i \in [q]$ let $T_i$ be a spanning tree for each $C_i$. Let $G^*$ be the component contracted graph of $G$ using the $T_i$'s. Let $R^* = \{c^i\vert i \in [q] \}$ be the terminal set of $G^*$. By, Observation~\ref{obs:terminal-comps}, if $G$ is a yes-instance then it must be that $q \leq \compconst k$. If this is not the case, then we immediately output no. 
 
 From now on, we are in the case $q \leq \compconst k$. By Observation~\ref{obs:contract}, it is enough to determine whether $(G^*,R^*,q,k)$ is a yes-instance of \stree. As noted earlier, $G^*$ may no longer be a clique-grid graph. 
 
 We run the Dreyfus-Wagner algorithm \cite{dreyfus1971steiner} which returns a minimum edge-weighted Steiner tree connecting $R^*$ in $G^*$. Since $G^*$ is unweighted, the returned solution Steiner tree $T$ has the minimum number of edges. Note that since $G^*$ is unweighted, a Steiner tree for $R^*$ minimizes the number of Steiner vertices if and only if it has minimum number of edges. The total number of Steiner vertices in $T$ is $\vert V(T) \vert - \vert R^*\vert$. If $\vert V(T) \vert - \vert R^*\vert \leq k$, then our algorithm returns that $(G^*,R^*,q,k)$ is a yes-instance of \stree, and otherwise it returns no.
 
The construction of $G^*$ is done in polynomial time. Since $q \leq \compconst k$, the Dreyfus-Wagner algorithm runs in $2^{O(k)}n^{O(1)}$. Thus, our algorithm also has running time $2^{O(k)}n^{O(1)}$. 
 \end{proof}

\section{W[1]-Hardness for \stree on Disk Graphs}\label{sec:whard}

In this section, we consider the \stree problem on disk graphs and prove that this problem is W[1]-hard parameterized by the number Steiner vertices $k$. 


\whardness*
\begin{proof}
We prove Theorem~\ref{w1-hard} by giving a parameterized reduction from the \textsc{Grid Tiling with}  $\ge$ problem which is known to be W[1]-hard\footnote{$k\times k$ \textsc{Grid Tiling with} $\ge$ problem is W[1]-hard, assuming ETH, cannot be solved in $f(k)n^{o(k)}$ for any function $f$} \cite{cygan2015parameterized}. 
In the \textsc{Grid Tiling with} $\ge$ problem,
we are given an integer n, a $k\times k$ matrix for an integer $k$ and a set of pairs $S_{ij}\subseteq [n]\times [n]$ of each cell. The objective is to find, for each $1\le i,j\le k$, a value $s_{ij}\in S_{ij}$ such that if $s_{ij}=(a,b)$ and $s_{i+1,j}=(a',b')$ then $a\ge a'$; if $s_{ij}=(a,b)$ and $s_{i,j+1}=(a',b')$ then $b\ge b'$. 

Let $I=(n,k,\mathcal{S})$ be an instance of the \textsc{Grid Tiling with} $\ge$.
We construct a set of unit disks $D$, that is divided into three sets of unit disks $D_1, D_2, D_3$; $D=D_1\uplus D_2\uplus D_3$. 
Each disk in $D_1, D_2, D_3$ is of radius $1$, $\delta$ and $\kappa$, respectively. We will define the value of $\delta$ and $\kappa$ shortly. 
The construction of the set $D=D_1\uplus D_2\uplus D_3$ will ensure that $D$ contains a \stree with $k^2$ Steiner vertices if and only if $I$ is a yes instance of \textsc{Grid Tiling with} $\ge$. Let $\epsilon = 1/n^{10}$, and $\delta =\epsilon/4$. Here, we point out that the value of $\kappa, \epsilon$ are independent of each other.
First, we move the cells away from each other, such that the horizontal (resp. vertical) distance between the left columns (resp. top rows) any two consecutive cell is $2+\epsilon$. 
Let 
$100\delta$ be the side of length of each cell. Then, we introduce diagonal chains of terminal disks into $D_3$ of radius $\kappa=\sqrt{2}(2+\epsilon - 100\delta)/1000$ to connect the cells diagonally; see Figure~\ref{W-hard-connected}(a). For every $1\le x, y\le k$, and every $(a,b)\in S(x,y)\subseteq [n]\times [n]$,
we introduce into $D_1$ a disk of radius $1$ centered at $(2x+\epsilon x+\epsilon a, 2y+\epsilon y+\epsilon b)$.
Let $D[x,y]\subseteq D_1$ be the set of 
disks introduced for a fixed $x$ and $y$, and notice that they mutually intersect each other. 
Next, for $1\le x, y\le k$, we introduce into $D_2$,  disks of radius $\delta$ between consecutive cells of coordinate $(2x+1+\epsilon x+\epsilon a, 2y+\epsilon y)$ (placed horizontally); and $(2x+\epsilon x, 2y+1+\epsilon y+\epsilon b)$ (placed vertically). For every cell $S[x,y]$,
we denote the top, bottom, left, right cluster of terminal disks of radius $\delta$ from $D_2$ by $L[x,y], R[x,y], T[x,y], B[x,y]$, respectively. 
Moreover, for each cell $S[x,y]$, we introduce a disk of radius $\delta$ 
at a coordinate that is completely inside the rectangle  bounding the centres of disks in $D[x,y]$. This is to enforce that at least one disk is chosen form each $D[x,y]$. 
See Figure~\ref{W-hard-connected}(b) for an illustration. 

\begin{figure}[h]
\centering
\includegraphics[scale=.8]{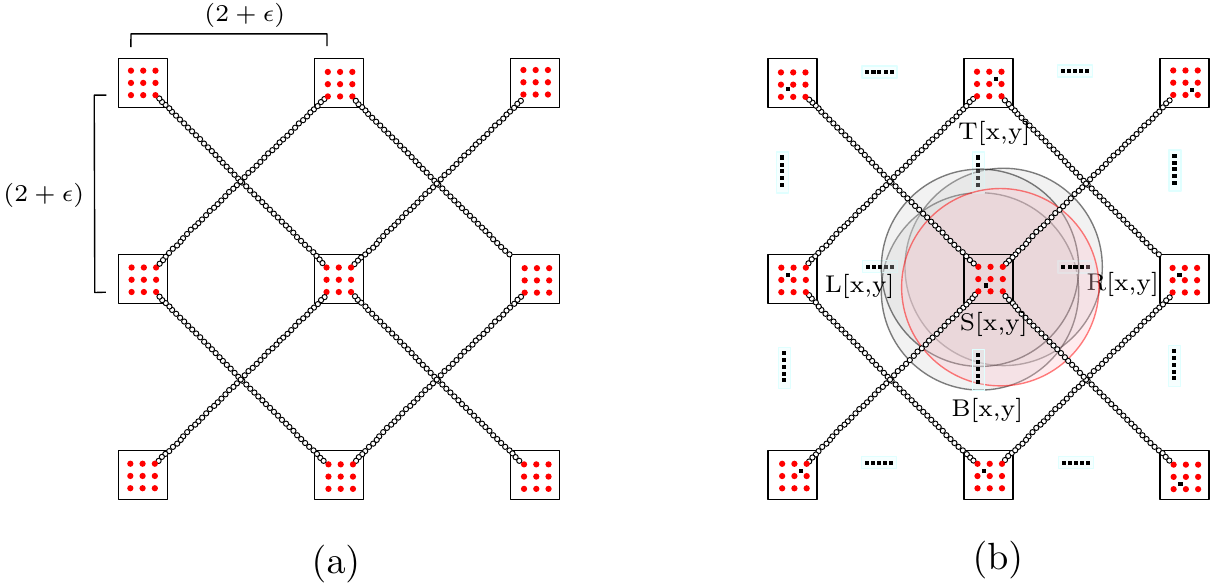}
\caption{(a) The schematic diagram of the cells, after adjusting the distance between adjacent cells which is $2+\epsilon$. The red disks inside each cells, are the coordinates where the center of the Steiner disk of radius~$1$  will be placed. The diagonal chains consisting of terminal disks of radius $\kappa$, are connecting the cells diagonally. (b) The small black dots inside each cell are extra terminals of radius $\delta$. 
Consider a cell $S[x,y]$. The shaded grey disks are the potential disks and the shaded red disk is chosen in the solution from $D[x,y]$.}
\label{W-hard-connected}
\end{figure}

We proceed with the following observation. 
Consider a disk $p$ that is centered at 
$(2x+\epsilon x+\epsilon a, 2y+\epsilon y+\epsilon b)$ for some $(a,b)\in [n]\times [n]$. Now, consider a disk $q$ from $R[x,y]$ centered at $(2x+1+\epsilon x+\epsilon a, 2y+\epsilon y)$. The distance between their centers are $\sqrt{1+\epsilon^2 b^2}$. 
We need to show that this is less than $(1+\epsilon/4)$. 
This is true because $1+\epsilon^2 b^2$ is less than $(1+\epsilon/4)$ as the value of $b$ goes to $n$, $\epsilon = 1/n^{10}$ and the value of $n$ is large. Hence, $q$ is covered by the disk $p$ from $S[x,y]$ centered at $(a,b)$. Next, consider a disk $q'$ from $R[x,y]$ centered at $(2x+1+\epsilon x+\epsilon (a+1),  2y+\epsilon y)$. The distance between their centers are $\sqrt{(1+\epsilon)^2+\epsilon^2 b^2}$. We show that this value is bigger than $(1+\epsilon/4)$. This means $(1+\epsilon)^2+\epsilon^2 b^2$ is bigger than 
$(1+\epsilon/4)^2$. As the value of $b$ goes to $n$, it is not hard to see the left side is bigger since $\epsilon=1/n^{10}$ and the value of $n$ is large. 
Therefore, $q'$ is not covered by the disk $p$ from $S[x,y]$ centered at $(a,b)$. The same calculation holds for $L[x,y]$, $T[x,y]$
and $B[x,y]$.

In the forward direction, let the pairs $s[x,y]\in S[x,y]$ form a solution for instance $I$, and let $s[x,y]=(a[x,y],b[x,y])$.
For every $1\le x,y\le k$, 
we select the disk $d[x,y]$ from $D_1$ of radius $1$ centered at 
$(2x+\epsilon x+\epsilon a[x,y], 2y+\epsilon y+\epsilon b[x,y])$. 
We have seen in the previous paragraph that this disk cover any disk from $R[x,y]$ of center with $(2x+1+\epsilon x+\epsilon a[x,y], 2y+\epsilon y)$ but does not covers disks with coordinate $(2x+1+\epsilon x+\epsilon (a[x,y]+1), 2y+\epsilon y)$. Similarly, this holds for $L[x,y],T[x,y],B[x,y]$. $s[x,y]$'s forms a solution of $I$, then we have $a[x,y]\ge a[x+1,y]$.  Therefore, the disks $d[x,y]$ and $d[x+1,y]$ will cover all disks from $R[x,y]$. Similarly, we have 
$b[x,y]\ge b[x,y+1]$ which implies that $d[x,y]$ and $d[x,y+1]$ will cover $T[x,y]$ and form a component them. Now, the diagonals chains consisting of terminal disks of radius $\kappa$, we have taken to join the cells (see Figure~\ref{W-hard-connected}(a)) ensures that all cells are connected. Moreover, we have shown that if $s[x,y]$'s form a solution of instance $I$, then all terminals in $L[x,y], R[x,y], T[x,y], B[x,y]$ (for any $1\le x, y\le k$) are covered. Therefore, this will form a connected Steiner tree with $k^2$ many Steiner disks.

In the reverse direction, 
let $D'\subseteq D_1$ be a set of $k^2$ Steiner disks that spans over all terminals in $D_2\cup D_3$.  This is true when for every $1\le x,y\le k$, the set $D'$ contains a disk $d[x,y]\in D[x,y]$ that is centered at $(2x+\epsilon x+\epsilon a[x,y], 2y+\epsilon y+\epsilon b[x,y])$ for some $(a[x,y],b[x,y]) \in [n]\times [n]$. Indeed, we are required to choose one disk from $D[x,y]$ due to the reason that there is a terminal disk lying inside the rectangle bounding the centres of disks in $D[x,y]$. The claim is that $s[x,y]=(a[x,y],b[x,y])$'s form a solution of $I$. First of all, $d[x,y]\in D[x,y]$ implies that $s[x,y]\in S[x,y]$. Consider a cell $S[x,y]$. We have observed that it covers disk $q$ from $R[x,y]$ centered at $(2x+1+\epsilon x+\epsilon a, 2y+\epsilon y)$, but a disk $q'$ from $R[x,y]$
centered at $(2x+1+\epsilon x+\epsilon (a+1),  2y+\epsilon y)$ is not covered. 
This is true for $L[x,y], T[x,y], B[x,y]$. Hence, if all terminals points from inside $S[x,y]$'s and $L[x,y], R[x,y], T[x,y], B[x,y]$ are covered by $k^2$ many Steiner disks, it would imply that  $a[x,y]\ge a[x+1,y]$ and $b[x,y]\ge b[x,y+1]$.
Therefore, $s[x,y]$'s form the solution for \textsc{Grid Tiling with} $\ge$ instance $I$. This completes the proof. 
\end{proof}

\paragraph*{Conclusion}
In this paper we studied the parameterized complexity of \stree on unit disk graphs and disk graphs under the parameterizations of $k$ and $t+k$. In future, we wish to explore tight bounds for the algorithms we have obtained and to probe into kernelization questions under these parameters. It would also be interesting to consider the minimum weight of a solution $k$-Steiner tree as a parameter. A variant of \stree that usually is easier to study is {\sc Full Steiner Tree}. However, in the case of unit disk graphs this problem proved to be very resilient to all our algorithmic strategies. We wish to explore {\sc Full Steiner Tree} on unit disk graphs under natural and structural parameters in future works.


\bibliography{main.bib}


\end{document}